\documentclass[english]{article}
\usepackage[T1]{fontenc}
\usepackage[utf8]{inputenc}
\usepackage{geometry}
\usepackage{color}
\usepackage{babel}
\usepackage{float}
\usepackage{amsmath}
\usepackage{amsthm}
\usepackage{setspace}
\usepackage[unicode=true,pdfusetitle,
 bookmarks=true,bookmarksnumbered=false,bookmarksopen=false,
 breaklinks=false,pdfborder={0 0 1},backref=false,colorlinks=false]
 {hyperref}

\makeatletter
\theoremstyle{plain}
\newtheorem{prop}{\protect\propositionname}

\makeatother

\providecommand{\propositionname}{Proposition}

\begin{document}
\begin{doublespace}
\begin{center}
\textbf{\large{}Are Instrumental Variables Really That Instrumental?
Endogeneity Resolution in Regression Models for Comparative Studies}{\large\par}
\par\end{center}

\begin{center}
\textbf{\large{}Ravi Kashyap (ravi.kashyap@stern.nyu.edu)}{\large\par}
\par\end{center}

\begin{center}
\textbf{\large{}Estonian Business School, Tallinn, Estonia / Formation
Fi, Hong Kong / City University of Hong Kong, Hong Kong}{\large\par}
\par\end{center}

\begin{center}
\begin{center}
\today
\par\end{center}
\par\end{center}

\begin{center}
{\large{}Keywords: Endogeneity; Bias; Coefficients; Instrumental Variable;
Comparative; Regression Model; Time Series Analysis; Econometric;
Experimental Design; Quantitative Methodology}{\large\par}
\par\end{center}

\begin{center}
{\large{}Journal of Economic Literature Codes: C32 Time-Series Models;
C36 Instrumental Variables (IV) Estimation; C44 Statistical Decision
Theory}{\large\par}
\par\end{center}

\begin{center}
{\large{}Mathematics Subject Classification Codes: 62M10 Time series,
regression; 91G70 Statistical methods, econometrics; 97D10 Comparative
studies}{\large\par}
\par\end{center}

\begin{center}
\textbf{\textcolor{blue}{\href{https://doi.org/10.5705/ss.202019.0459}{Edited Version: Kashyap, R. (2022).  Are Instrumental Variables Really That Instrumental? Endogeneity Resolution in Regression Models for Comparative Studies. Statistica Sinica, 32 (2022), 645-651, }}}\tableofcontents{}
\par\end{center}
\end{doublespace}
\begin{doublespace}

\section*{{\large{}Abstract }}
\end{doublespace}

\begin{doublespace}
\noindent {\large{}We provide a justification for why, and when, endogeneity
will not cause bias in the interpretation of the coefficients in a
regression model. This technique can be a viable alternative to, or
even used alongside, the instrumental variable method. We show that
when performing any comparative study, it is possible to measure the
true change in the coefficients under a broad set of conditions. Our
results hold, as long as the product of the covariance structure between
the explanatory variables and the covariance between the error term
and the explanatory variables are equal, within the same system at
different time periods or across multiple systems at the same point
in time.}{\large\par}
\end{doublespace}
\begin{doublespace}

\section{{\large{}\label{sec:Introduction}Introduction}}
\end{doublespace}

\begin{doublespace}
\noindent {\large{}Endogeneity is an econometric issue that arises
in regression models, causing the coefficients of the independent
(or explanatory) variables to be biased. Despite the problem having
been studied extensively in the literature, many researchers are either
unaware of the extent of the errors caused by this issue, or they
end up looking for solutions that might not be easily found (Avery
2005; Kim \& Washington 2006; Gordon 2015; Hill et al. 2021).}{\large\par}

\noindent {\large{}Consider the multiple linear regression model,
shown using matrix notation, in Equation (\ref{eq:Matrix_Regression}).
Here, $\boldsymbol{y}$ is a vector of observed values of the dependent
variables, $\boldsymbol{X}$ is a matrix of the independent variables,
$\boldsymbol{\beta}$ is the parameter vector to be estimated, and
$\boldsymbol{u}$ is a vector of the error terms:
\begin{equation}
\boldsymbol{y}=\boldsymbol{X\beta}+\boldsymbol{u}.\label{eq:Matrix_Regression}
\end{equation}
Endogeneity arises when the error term in the regression model and
the independent variables are correlated, that is, if }$\boldsymbol{E\left[X^{T}u\right]\neq0}$. 

\noindent {\large{}We are interested in studying this model (Equation
\ref{eq:Matrix_Regression}), with $\boldsymbol{m}$ and $\boldsymbol{n}$
independent and identically distributed (i.i.d) samples, before and
after an event. For such an event study, let the true coefficients
before and after the event be denoted by $\boldsymbol{\beta_{B}}$
and $\boldsymbol{\boldsymbol{\beta_{A}}}$, respectively. Clearly,
$\boldsymbol{\beta_{B}}$ and $\boldsymbol{\boldsymbol{\beta_{A}}}$
could be the coefficients for two different systems we are looking
to compare. We wish to make inferences on the difference between the
coefficients, $\boldsymbol{\beta_{B}-\boldsymbol{\beta_{A}}}$, under
the possible presence of endogeneity.}{\large\par}

\noindent {\large{}Our main result (Proposition \ref{prop:Event_Study-Endogeneity-Resolution})
in Section (\ref{sec:CER}) gives a mathematical justification for
the conditions under which endogeneity will not cause significant
bias in the interpretation of the regression coefficients. Even if
these criteria are not met in their entirety, we will see that endogeneity
is mitigated to a great extent. This illustrates that our solution
is an ideal way to approach the construction of regression models.
Despite the simplicity of our methodology, to the best of our knowledge,
no other works in the literature employ a similar approach to either
partially combat or completely resolve endogeneity. The relative ease
of implementation should allow researchers from almost any discipline
to effortlessly adopt our solution when designing their experiments
or while utilizing econometric techniques on the data collected. Section
(\ref{sec:EERI}) considers situations under which our results might
hold, including suggestions for further topics that need to be investigated
to make this approach more useful. Section (\ref{sec:DEM}) discusses
the main consequences of our result and concludes the paper.}{\large\par}
\end{doublespace}
\begin{doublespace}

\section{{\large{}\label{sec:CER}Condition for Endogeneity Resolution}}
\end{doublespace}

\begin{doublespace}
\noindent {\large{}With the framework given in Equation (\ref{eq:Matrix_Regression}),
we present our primary result below.}{\large\par}
\end{doublespace}
\begin{prop}
\begin{doublespace}
\noindent {\large{}\label{prop:Event_Study-Endogeneity-Resolution}Let
the true coefficients of the regression model be denoted by $\boldsymbol{\beta_{B}}$
and $\boldsymbol{\boldsymbol{\beta_{A}}}$. The suffixes $\boldsymbol{B}$
and $\boldsymbol{A}$ denote values before and after an event, respectively,
or across the two different systems under comparison. Suppose we have
not been able to completely eliminate endogeneity. Then, the biased
coefficients are denoted by $\boldsymbol{\hat{\beta}_{BE}}$ and $\boldsymbol{\hat{\beta}_{AE}}$,
and $\boldsymbol{c}$ is a vector that holds the covariance of the
error term with each of the explanatory variables. The change in the
coefficient values before and after the event or across the two systems
is given by,
\begin{equation}
\boldsymbol{E\left[\hat{\beta}_{AE}\right]}-\boldsymbol{E\left[\hat{\beta}_{BE}\right]}=\boldsymbol{\boldsymbol{\beta_{A}}-\boldsymbol{\beta_{B}}}\boldsymbol{+}\boldsymbol{E\left[\left(X_{A}^{T}X_{A}\right)^{-1}\boldsymbol{c_{A}}\right]}\boldsymbol{-}\boldsymbol{\left[E\left(X_{B}^{T}X_{B}\right)^{-1}\boldsymbol{c_{B}}\right]}.
\end{equation}
Because we are considering only the difference in the coefficients,
we end up measuring the true change in the coefficients, as long as
\begin{equation}
\boldsymbol{E\left[\left(X_{B}^{T}X_{B}\right)^{-1}\boldsymbol{c_{B}}\right]}=\boldsymbol{E\left[\left(X_{A}^{T}X_{A}\right)^{-1}\boldsymbol{c_{A}}\right]}.\label{eq:Equality-Criterion}
\end{equation}
}{\large\par}
\end{doublespace}
\end{prop}
\begin{proof}
\begin{doublespace}
\noindent {\large{}The result follows by taking the difference and
then the expectations of the coefficient estimators when there is
endogeneity.}{\large\par}
\end{doublespace}
\end{proof}
\begin{doublespace}
\noindent {\large{}What our result suggests is that as long as the
product of the covariance structure between the explanatory variables
(inverse of the covariance matrix) and the covariance between the
error term and the explanatory variables have equal expectations before
and after the event (or across the two systems), we end up measuring
the actual change between the coefficients. This ensures that we are
using the correct values to understand how any system has changed
across time, or how two different systems can be evaluated at the
same point in time. This condition is trivially satisfied when $\boldsymbol{E\left[\left(X_{B}^{T}X_{B}\right)^{-1}\boldsymbol{c_{B}}\right]}=\boldsymbol{E\left[\left(X_{A}^{T}X_{A}\right)^{-1}\boldsymbol{c_{A}}\right]}=0$. }{\large\par}
\end{doublespace}
\begin{doublespace}

\section{{\large{}\label{sec:EERI}Embarking on Endogeneity Resolution Improvements}}
\end{doublespace}

\begin{doublespace}
\noindent {\large{}Endogeneity can occur under three scenarios: 1)
if the dependent variables can influence the explanatory variables,
and vice versa; 2) some key variables are omitted in the regression
model; and 3) there are errors in the measurement of the variables.
Endogeneity arises in almost any study, as a silent malice, because
it is difficult to completely rule out omitted variables and measurement
errors. The first scenario, in which the dependent variables can influence
the explanatory variables, and vice versa, is generally observable.
Most conventional models require additional information for the coefficients
to be consistently estimated.}{\large\par}

\noindent {\large{}The most common workaround for endogeneity is to
use an additional variable that is uncorrelated with the error term
(instrument exogeneity), but is correlated with the explanatory variable
(instrument relevance) in the original model. This is known as the
instrumental variables regression method (Cragg \& Donald 1993; Verbeek
2008; Semadeni, Withers \& Trevis Certo 2014). The main difficulty
with this option is locating valid instruments that satisfy the conditions
of relevance and exogeneity (Miguel, Satyanath \& Sergenti 2004; Baser
2009; Sovey \& Green 2011). When the conditions from our result are
satisfied, we do not have to look for instrumental variables, which
in many instances are not easy to find.}{\large\par}

\noindent {\large{}Below, we list lines of investigation we are currently
working on. These can help to make our result more robust as well
as shed further light on the scenarios under which it would be more
applicable.}{\large\par}
\end{doublespace}
\begin{enumerate}
\begin{doublespace}
\item {\large{}Based on the type of dependency, between the covariance of
the error term with each of the explanatory variables and the covariance
structure between the explanatory variables, endogeneity will no longer
be an issue. For example, using Equation (\ref{eq:Equality-Criterion}),
we could have that $\boldsymbol{E\left[\left(X_{B}^{T}X_{B}\right)^{-1}\boldsymbol{c_{B}}\right]}=$
$\boldsymbol{E\left[\left(X_{B}^{T}X_{B}\right)^{-1}\right]E\left[\boldsymbol{c_{B}}\right]}$
and $\boldsymbol{\boldsymbol{E\left[\left(X_{A}^{T}X_{A}\right)^{-1}\boldsymbol{c_{A}}\right]=}}$
$\boldsymbol{E\left[\left(X_{A}^{T}X_{A}\right)^{-1}\right]E\left[\boldsymbol{c_{A}}\right]}$.
The requirement from our main result holds if the expectations are
equal individually, $\boldsymbol{E\left[\boldsymbol{c_{B}}\right]}=\boldsymbol{E\left[\boldsymbol{c_{A}}\right]}$
and $\boldsymbol{E\left[\left(X_{B}^{T}X_{B}\right)^{-1}\right]}=\boldsymbol{E\left[\left(X_{A}^{T}X_{A}\right)^{-1}\right]}$.
Even if the expectations of the variables are not equal, as long as
the variations in one of them is countered by variations in the other,
so that they are equal when taken together, as given by the expression
$\boldsymbol{E\left[\left(X_{B}^{T}X_{B}\right)^{-1}\right]E\left[\boldsymbol{c_{B}}\right]}=\boldsymbol{E\left[\left(X_{A}^{T}X_{A}\right)^{-1}\right]E\left[\boldsymbol{c_{A}}\right]}$,
our results will hold. Many other dependency structures are possible
and form fascinating avenues for further research.}{\large\par}
\item {\large{}It would be a fruitful line of inquiry to study Equation
(\ref{eq:Equality-Criterion}) and its related properties under each
of the three cases under which endogeneity arises, and to derive expressions,
including tests, that indicate whether our results hold for each situation.}{\large\par}
\item {\large{}It is quite possible that equality might turn out to be a
difficult criterion to satisfy. It would be helpful to show that there
might be convergence of the relevant parameter values under various
assumptions on the distributions. Specifically, we would examine the
convergence of the left and right sides of Equation (\ref{eq:Equality-Criterion}).
If convergence results cannot be achieved for some cases, depending
on the empirical context, approximations can be performed keeping
some error limits in mind. These could be upper or lower bounds on
the difference between the two sides of Equation (\ref{eq:Equality-Criterion}). }{\large\par}
\item {\large{}In some instances, it is likely that the true coefficients
might be equal, $\boldsymbol{\beta_{B}=\boldsymbol{\beta_{A}}}$.
Any test to verify the equality of the coefficients requires the condition
in Equation (\ref{eq:Equality-Criterion}), because we are taking
the differences and then the expectations of the coefficient estimators
to arrive at this result. This is necessary even if the mean difference
between the coefficient estimators is zero.}{\large\par}
\item {\large{}When comparing two different systems, it is quite clear how
we should estimate the relevant variables in the left and right sides
of Equation (\ref{eq:Equality-Criterion}) . However, while appraising
the same system across different points, the concepts of before and
after need clarification. This involves coming up with ways to split
the data sample. The context under which the data have been collected
should provide some guidelines, but mathematical suggestions on how
to do this should not be ruled out, and can be pursued further.}{\large\par}
\item {\large{}A simulation-based study to justify the theory we have put
forward would be worthwhile. Supplementing such a simulation-based
study with actual data would be a good next step.}{\large\par}
\end{doublespace}
\end{enumerate}
\begin{doublespace}
\noindent {\large{}Much work can be done in this space toward developing
a stronger theoretical foundation supporting the effectiveness of
this approach under relatively general assumptions on the structure
of endogeneity, or in terms of coming up with tests that can identify
whether certain assumptions hold before starting the analysis. Needless
to say, some subjective decisions have to be made, but that is the
case with all econometric modeling, testing, and interpretation.}{\large\par}
\end{doublespace}
\begin{doublespace}

\section{{\large{}\label{sec:DEM}Does Endogeneity Matter? Endogeneity Does
Matter, But Only Sometimes!}}
\end{doublespace}

\begin{doublespace}
\noindent {\large{}We have considered the issue of endogeneity in
regression models. We have aimed to keep our paper as concise as possible,
yet with sufficient detail to keep it self contained, so that the
results are immediately useful to researchers across different disciplines.
We have shown that endogeneity will not cause major bias in the interpretation
of the coefficients, as long as the product of the covariance structure
between the explanatory variables and the covariance between the error
term and the explanatory variables are equal in magnitude, within
the same system at different points in time, or across different systems
at the same point in time. }{\large\par}

\noindent {\large{}When performing any comparative study (Clasen 2004;
Hantrais 2008), the conditions we have outlined for the resolution
of endogeneity are usually satisfied. This suggests that if we are
not able to find valid instrumental variables, which is the most commonly
used approach when endogeneity is suspected to be present, we need
to consider designing our study as a comparative analysis. Even if
we have reasonably good instruments, we can have a comparative angle
in our study to supplement the results because this alleviates the
effect of omitted variables and errors in measurement, which are usually
difficult to detect. }{\large\par}

\noindent {\large{}Clearly, a comparative study can assess the same
system at different time periods, or it can be across multiple systems
at the same point in time. The main prescription from our result is
that, by conducting a comparative analysis, we can measure the true
change in the coefficients. Another justification for designing studies
with a comparative angle is that having an anchor point, or a frame
of reference for comparison, results in better decision making, as
opposed to making absolute judgments in isolation (Kahneman 1992;
Gavirneni \& Xia 2009).}{\large\par}

\noindent {\large{}Despite endogeneity having been extensively studied,
and numerous complicated methods having been put forth to deal with
it, simple alternatives such as those presented here would be extremely
useful for researchers. This will allow better use of the data they
have collected from their experiments, and even the design of better
experiments so that the resulting data can be analyzed with minimal
errors. Our simple workaround works quite well under a broad set of
conditions, especially when valid instruments are difficult to find.
More importantly, even if endogeneity is not completely removed, owing
to suitable conditions not being satisfied, it can be partly eliminated
in a much broader set of cases.}{\large\par}
\end{doublespace}
\begin{doublespace}

\section{{\large{}Acknowledgements}}
\end{doublespace}

\begin{doublespace}
\noindent {\large{}Dr. Yong Wang, Dr. Isabel Yan, Dr. Vikas Kakkar,
Dr. Fred Kwan, Dr. William Case, Dr. Srikant Marakani, Dr. Qiang Zhang,
Dr. Costel Andonie, Dr. Jeff Hong, Dr. Guangwu Liu, Dr. Humphrey Tung
and Dr. Xu Han at the City University of Hong Kong; Dr. Rong Chen,
anonymous reviewers and numerous seminar participants from various
events provided many suggestions to improve this paper. The views
and opinions expressed in this article, along with any mistakes, are
mine alone and do not necessarily reflect the official policy or position
of either of my affiliations or any other agency.}{\large\par}
\end{doublespace}
\begin{doublespace}

\section{{\large{}References}}
\end{doublespace}
\begin{enumerate}
\begin{doublespace}
\item {\large{}Avery, G. (2005). Endogeneity in logistic regression models.
Emerging infectious diseases, 11(3), 503.}{\large\par}
\item {\large{}Baser, O. (2009). Too much ado about instrumental variable
approach: is the cure worse than the disease?. Value in health, 12(8),
1201-1209.}{\large\par}
\item {\large{}Clasen, J. (2004). Defining comparative social policy. A
handbook of comparative social policy, 91-102.}{\large\par}
\item {\large{}Cragg, J. G., \& Donald, S. G. (1993). Testing identifiability
and specification in instrumental variable models. Econometric Theory,
9(2), 222-240.}{\large\par}
\item {\large{}Gavirneni, S., \& Xia, Y. (2009). Anchor selection and group
dynamics in news-vendor decisions—A note. Decision Analysis, 6(2),
87-97.}{\large\par}
\item {\large{}Gordon, D. V. (2015). The endogeneity problem in applied
fisheries econometrics: A critical review. Environmental and Resource
Economics, 61(1), 115-125.}{\large\par}
\item {\large{}Hill, A. D., Johnson, S. G., Greco, L. M., O’Boyle, E. H.,
\& Walter, S. L. (2021). Endogeneity: A review and agenda for the
methodology-practice divide affecting micro and macro research. Journal
of Management, 47(1), 105-143.}{\large\par}
\item {\large{}Hantrais, L. (2008). International comparative research:
Theory, methods and practice. Macmillan International Higher Education.}{\large\par}
\item {\large{}Kim, D. G., \& Washington, S. (2006). The significance of
endogeneity problems in crash models: an examination of left-turn
lanes in intersection crash models. Accident Analysis \& Prevention,
38(6), 1094-1100.}{\large\par}
\item {\large{}Kahneman, D. (1992). Reference points, anchors, norms, and
mixed feelings. Organizational behavior and human decision processes,
51(2), 296-312.}{\large\par}
\item {\large{}Miguel, E., Satyanath, S., \& Sergenti, E. (2004). Economic
shocks and civil conflict: An instrumental variables approach. Journal
of political Economy, 112(4), 725-753.}{\large\par}
\item {\large{}Semadeni, M., Withers, M. C., \& Trevis Certo, S. (2014).
The perils of endogeneity and instrumental variables in strategy research:
Understanding through simulations. Strategic Management Journal, 35(7),
1070-1079.}{\large\par}
\item {\large{}Sovey, A. J., \& Green, D. P. (2011). Instrumental variables
estimation in political science: A readers’ guide. American Journal
of Political Science, 55(1), 188-200.}{\large\par}
\item {\large{}Verbeek, M. (2008). A guide to modern econometrics. John
Wiley \& Sons.}{\large\par}
\end{doublespace}
\end{enumerate}

\end{document}